\documentclass[reqno,11pt]{amsart}
\usepackage{amsmath, latexsym, amsfonts, amssymb}
\usepackage{mathrsfs,enumerate}
\usepackage{graphics,epsfig,color}

\numberwithin{equation}{section}

\newtheorem{theorem}{Theorem}[section]

\newtheorem{lemma}[theorem]{Lemma}
\newtheorem{prop}[theorem]{Proposition}

\newcommand{\cE}{{\ensuremath{\mathcal E}} }

\newcommand{\cC}{{\ensuremath{\mathcal C}} }
\newcommand{\cN}{{\ensuremath{\mathcal N}} }
\newcommand{\cO}{{\ensuremath{\mathcal O}} }

\newcommand{\cR}{{\ensuremath{\mathcal R}} }

\newcommand{\bbE}{{\ensuremath{\mathbb E}} }

\newcommand{\bbN}{{\ensuremath{\mathbb N}} }
\newcommand{\N}{{\ensuremath{\mathbb N}} }

\newcommand{\bbP}{{\ensuremath{\mathbb P}} }

\newcommand{\bbR}{{\ensuremath{\mathbb R}} }

\newcommand{\bbZ}{{\ensuremath{\mathbb Z}} }

\newcommand{\bo}{{\ensuremath{\mathbf 1}} }

\newfont{\indic}{bbmss12}

\title[ Fick's law in a random lattice Lorentz gas ]
      {Fick's law in a random lattice Lorentz gas}

\author[R.\ Lefevere]{Rapha\"el Lefevere}
 \address{Laboratoire de Probabilit\'es
  et Mod\`eles Al\'eatoires (CNRS UMR 7599), Universit\'e Paris Diderot,
UFR de Math\'ematiques, b\^atiment Sophie Germain,
5 rue Thomas Mann,
75205 Paris CEDEX 13
France}
\email{lefevere\@@math.univ-paris-diderot.fr}
\begin{document}

\begin{abstract}
We provide a proof that the stationary macroscopic current of particles in a random lattice Lorentz gas satisfies Fick's law when connected to particles reservoirs. We consider a box on a $d+1$-dimensional lattice and when $d\geq7$, we show that under a diffusive rescaling of space and time, the probability to find a current different from its stationary value is exponentially small in time. Its stationary value is given by the conductivity times the difference of chemical potentials of the reservoirs. The proof is based on the fact that in high dimension, random walks have a small probability of making loops or intersecting each other when starting sufficiently far apart.

\end{abstract}

\maketitle
\section{Introduction}
Ever since the works of the founding fathers of statistical mechanics, the derivation of the macroscopic laws of physics as the result of the motion of the microscopic components has been a major challenge which remains largely unsolved to this day.    Fick's law is one of those central laws of macroscopic physics.  It states that, after some transient time, the current of particles crossing an extended macroscopic system of length $L$ decreases like the inverse power of $L$.  A paradigmatic model  in this context is provided by the Lorentz gas : it consists of tracer particles moving freely in a box and colliding with fixed obstacles.  The only rigorous derivation of Fick's law was achieved by Bunimovich and Sinai in \cite{bunisinai} for a finite horizon Lorentz gas when the scatterers have a specific shape that gives rise to a strongly chaotic dynamics. However, it is unlikely that the microscopic dynamics of a typical material possess the  special properties of a strongly chaotic billiard.
And, as advocated by Bunimovich \cite{bunimovich}, a more satisfactory result from a conceptual point of view would be to establish diffusion in a {\it random} Lorentz gas.  In that case, obstacles of arbitrary shape are thrown at random in a box.  The goal is to show that, after a diffusive rescaling of space and time is performed, macroscopic observables obey the laws of diffusion with very large probability with respect to the distribution of the obstacles.    If one looks at Fick's law for the macroscopic current, this requires to control not only  its average but also, at least, its variance.  In contrast to the the Bunimovich-Sinai case, the randomness of the scatterers induce correlations between the trajectories and therefore also between occupation numbers (or local empirical densities) at different points in space.

In this paper we consider the $d$-dimensional version of the model \cite{Lefevere} which can be seen as a random lattice Lorentz gas (see figure 3 below) introduced by Ruijgrook and Cohen \cite{Ruijgrok}.  In that model, also called the mirrors model, the motion of particles is restricted to the edges of a regular lattice and some scatterers sit randomly at the vertices of the lattice. The motion of the particles is described as a deterministic walk in a random environment.  

We couple our model with a fixed density of scatterers (i.e. a non-dilute gas) to particle baths at constant chemical potentials.   We focus on the macroscopic current of particles through a section of the system. In high dimensional systems, we establish Fick's law for the stationary current as a weak law of large number in the size of the system.  We also show that, under a diffusive rescaling of time,  at any time, the difference between the current and its stationary value is exponentially small in time.

The approach of this paper to diffusion in Lorentz gases is novel and different from the traditional one, based on the Boltzmann equation, used for instance in the recent paper  by Basile, Nota, Pezzotti and Pulvirenti\cite{Basile}.  There, it is shown that the stationary average current in the dilute case obeys Fick's law.  

The basic idea in our approach is the following. We first relate the macroscopic current to the orbits of the dynamical system.  This is a result that is valid for deterministic realisations of the scatterers.  In a second step we take a random distribution of scatterers. The orbits become then random objects similar to random loops with strong exclusion constraints among themselves.  We show that in high dimension, most of the orbits cross the macroscopic system on a diffusive time-scale as if they were independent random walks.  This allows to obtain Fick's law as a weak law of large numbers. The average stationary current in this limit maybe identified as the difference between chemical potential times the probability for a particle to cross the system,  an idea that was put forward by Casati, Mejia-Monasterio and Prosen in \cite{Casati}, in the context of  chaotic systems.
We show that in high dimension,  the dominant part of this crossing probability is given by the probability that a lazy random walk crosses a system of size $N$.   This is possible partly because orbits do not make ``loop" of size smaller than $N$. The probability of jumping to a neighbour  in this walk gives the diffusion constant.
We have chosen to focus here on the macroscopic current and the approach to its stationary value.  The derivation of the diffusion equation  for the macroscopic density of particles as a law of large numbers in a diffusive scaling limit  is tractable by the same methods.  

There is no fundamental obstacle to apply the general strategy underlying this paper to the case of more general mirrors models and of continuous space and time dynamics. This is especially true for the results of section 3. Regarding the cases of lower dimensional versions of our model and of the mirrors models, we note that a more refined analysis  of the structure of the orbits is possible.    Also, it should be possible to take into account systematically the ``loops" and ``collision" between orbits.  This should lead to a renormalised diffusion constant different from the one that we obtain, which is directly proportional to the scatterers density.  The distribution of  the set of orbits could also be analysed by using a connection with random loops models appearing in the context of quantum spin systems, see the paper by Ueltschi \cite{Ueltschi}.  The results of Lacoin \cite{Lacoin} on random adjacent transpositions may also provide interesting results for our our model.

In section 2, we define  the dynamics of our model and state our main result. In section 3, we relate the current to the number of  orbits crossing the system.  In section 4 we prove some results on lazy random walks that are the keys of our analysis.  Basically, we first show  that in high dimension it is more unlikely for a (lazy) random walk to make a ``loop" of length $N$ than to cross a system of size $N$.  We also give an estimate on the probability that two (lazy) random walks intersect each other before exiting the system. Section 5 is devoted to the connection between the orbits of our dynamics and lazy random walks.  In section 6, we establish a law of large numbers for the number of orbits that crosses the system by using the estimates of section 5.
In the final section, we put the different parts together and prove our main result : Fick's law as a weak law of large numbers in the size of the system and the exponential approach to stationarity.

\section{Definition of the dynamics and main result}

We first define the $d$-dimensional version of the rings model of \cite{Lefevere}. 
Let us consider the $d$-dimensional box :
$$
\Lambda=\{1,\ldots,N\}^d=\{i=(i_1,\ldots,i_d), i_l\in \{1,\ldots,N\},\;1\leq l\leq d\}.
$$
To each site of $i\in\Lambda$, we attach a ring $\cR_i$ carrying $N$ sites $k\in\{1,\ldots,N\}$.
The model consists of particles moving on 
$$
\cC=\prod_{i\in\Lambda}\cR_i=\{(k,i):k\in\{1,\ldots,N\},\;i\in\Lambda\}.
$$
Rings are periodic : in the following, addition and substraction on the first component of points of $\cC$ are to be understood modulo $N$. We also impose periodic boundary conditions on the $d-1$ first components of points of $\Lambda$. Thus,  addition and substraction on those components  are also to be understood modulo $N$.  We define also $\Lambda_n=\{i'\in\bbZ^d: \exists i\in\Lambda, |i_d-i'_d|\leq n\}$. The distance between two points $i=(i_1,\ldots,i_d)$ and $i'=(i_1',\ldots,i_d')$ in $\Lambda_2$ is defined as follows : 
\begin{equation}
d(i,i')=\inf_{\substack{
j\in\bbZ^d\\ j_d=0}}\|i-i'+jN\|
\label{distance}
\end{equation}
where $\|i\|=\sum_{l=1}^d|i_l|$.
In the following, we shall refer to the first component of $x=(k,i)$ as the {\it vertical} component and to the components $i=(i_1,\ldots,i_d)$ as the {\it horizontal} ones.  We define boundaries of the system :
\begin{equation}
B_{-}=\{(k,i)\in\cC: i_d=1\}\quad {\rm and}\quad B_{+}=\{(k,i)\in\cC: i_d=N\} .
\end{equation}
$B=B_-\cup B_+$.
We denote by $(e_1,\ldots,e_d)$ the canonical basis of $\bbR^d$.
The second ingredient of the model is the presence of scatterers that are located in-between pairs of nearest-neighbours of the form $(k,i)$ and $(k,j)$ with $d(i,j)=1$. We define variables $\xi(k,ij)$ taking values in $\{0,1\}$ such that  $\xi(k,ij)=1$ if and only if there is a scatterer between sites $(k,i)$ and $(k,j)$, with $d(i,j)=1$. We use the notation :
$$
\xi=\{\xi(k,ij): k\in\{1,\ldots,N\}, i,j\in\Lambda_2, d(i,j)=1\}.
$$
Throughout the paper, we will set
\begin{equation}
\xi(k,ij)=1,\quad {\rm if }\;i\;{\rm or}\; j\notin \Lambda_2.
\label{xicon}
\end{equation}
The dynamics of the model is defined by the following dynamical system $F:\cC\to\cC$ :  for any $(k,i)\in \cC$,
\begin{equation}
F(k,i)=\sum_{j:d(i,j)=1}c(k,ij)(k+1,j)+(k+1,i)\prod_{j:d(i,j)=1}(1-c(k,ij))
\label{tau_def}
\end{equation}
where the sum and product run over $j\in\Lambda_1$ and
\begin{equation}
c(k,ij)=\xi(k,ij)\prod_{
\substack{
 l:d(i,l)=1\\
l\neq j
}
}
(1-\xi(k,il))
\prod_{
\substack{
 l:d(j,l)=1\\
l\neq i
}
}
(1-\xi(k,jl)).
\label{c_def}
\end{equation}
where the product runs over $l\in\Lambda_2$.
The condition (\ref{xicon}) ensures that $F$ is well defined from to $\cC$ to $\cC$.
We define the orbit of a point of $\cC$ and its period by 
\begin{equation}
{\mathcal O}(x)=\{y\in\cC: \exists t\geq 0,F^t(x)=y\}
\label{orbit}
\end{equation}
and
\begin{equation}
T(x)=\inf\{t\geq 0: F^t(x)=x\}.
\label{period}
\end{equation}

We gather here the following useful facts about the dynamical system.

\begin{lemma} 

\begin{enumerate}
\item $F$ is a well-defined bijective map from $\cC$ into $\cC$,
\item For every $x\in\cC$, ${\mathcal O}(x)$ is a loop :  $T(x)\leq |\cC|=N^{d+1}$,
\item Every orbit is self-avoiding  : for any $y\in {\mathcal O}(x)$, and $\forall t< T(x)$, $F^t(y)\neq y$,
\item Orbits are non-intersecting : if $ y\notin {\mathcal O}(x)$, then ${\mathcal O}(x)\cap{\mathcal O}(y)=\emptyset$.
\label{taumap}
\end{enumerate}
\end{lemma}
{\bf Remark} From $\cC$, we could build an extended phase space $\hat \cC$ to include a ``velocity" degree of freedom $\xi\in\{-1,1\}$.  We define a dynamical system $\hat F$ on $\hat\cC$ as follows.  We write any point of the new phase space $\hat\cC$ like $(x,\xi)$.
$$
\hat F(x,\xi)=\left\{\begin{array}{ll}
(F(x),\xi),\quad \xi=1\\
(F^{-1}(x),\xi),\quad \xi=-1\\
\end{array}
\right.
$$ 
It is easy to see that $\hat F$ has the same properties on $\hat\cC$ than the ones stated in Lemma \ref{taumap}. Moreover, $\hat F$ is reversible in the same sense than Hamiltonian dynamics is. Namely, if the reversal of velocity $\Pi$ is defined as $\Pi(x,\xi)=(x,-\xi)$, then 
$$
\hat F^{-1}=\Pi \hat F\Pi.
$$
In the sequel, to keep notations simple, we work only with the dynamical system $F$.
\vspace{2mm}

We put particles at the sites of $\cC$ in a such a way that a site carries at most one particle.  If this property is true at time zero, then it remains true for all time. 
This is not essential however, and one could allow more than one particle per site but it simplifies a bit the set-up.
 A variable $\sigma(x;t)=\sigma(k,i;t)\in\{0,1\}$ describes the state of occupation of the site $x=(k,i)\in\cC$ at time $t$.  $\sigma(\cdot;t)\in\{0,1\}^{\cC}$ denotes the configuration of occupation variables at time $t$.  The motion of a single particle is described as follows  :  if the particle is located at $(k,i)$ at time $t$, then at time $t+1$,  it moves according to the  dynamical system $F$, namely, it jumps to $F(k,i)$.
In other words, a particle located at site  $(k,i)$, i.e. at site $k$ on the ring $\cR_i$  will jump to site $k+1$ on ring $\cR_{j}$ (with $d(i,j)=1$), if and only if the following conditions are simultaneously satisfied :
\begin{enumerate}
\item There is  a scatterer between $(k,i)$ and $(k,j)$, namely $\xi(k,ij)=1$,
\item There are no other scatterers around that pair. 
\end{enumerate}
\begin{figure}[thb]
\includegraphics[width = .89\textwidth]{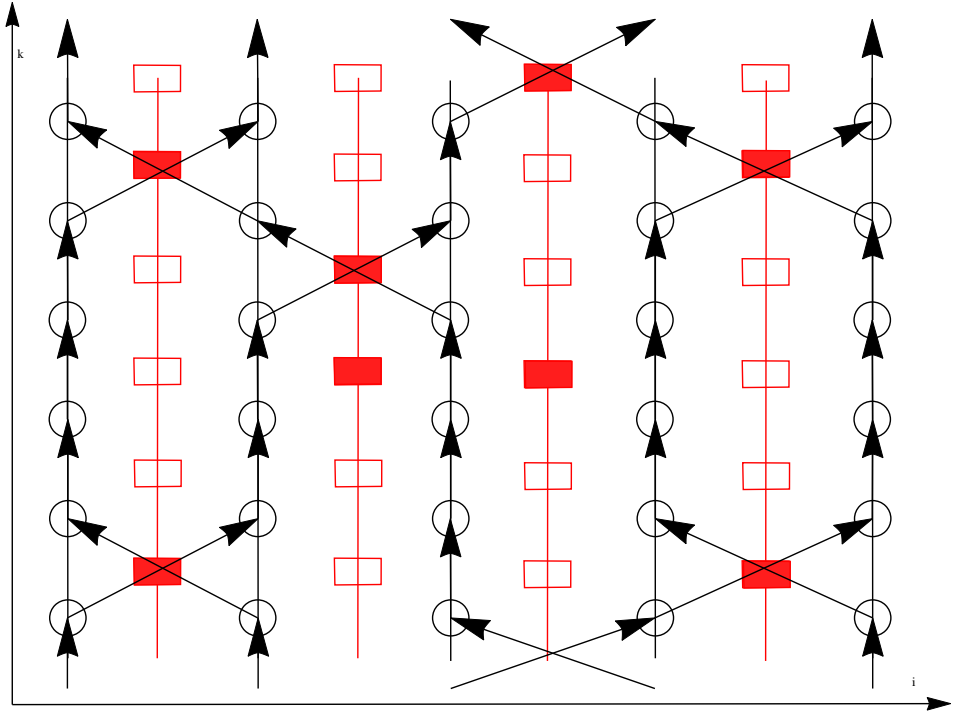} 
\caption{The dynamical system in 1D, the $k$ index corresponds to the vertical coordinate and the $i$ index corresponds to the horizontal one. Periodic boundary conditions are imposed on the vertical direction.}
\end{figure}
\begin{figure}[thb]
\includegraphics[width = .89\textwidth]{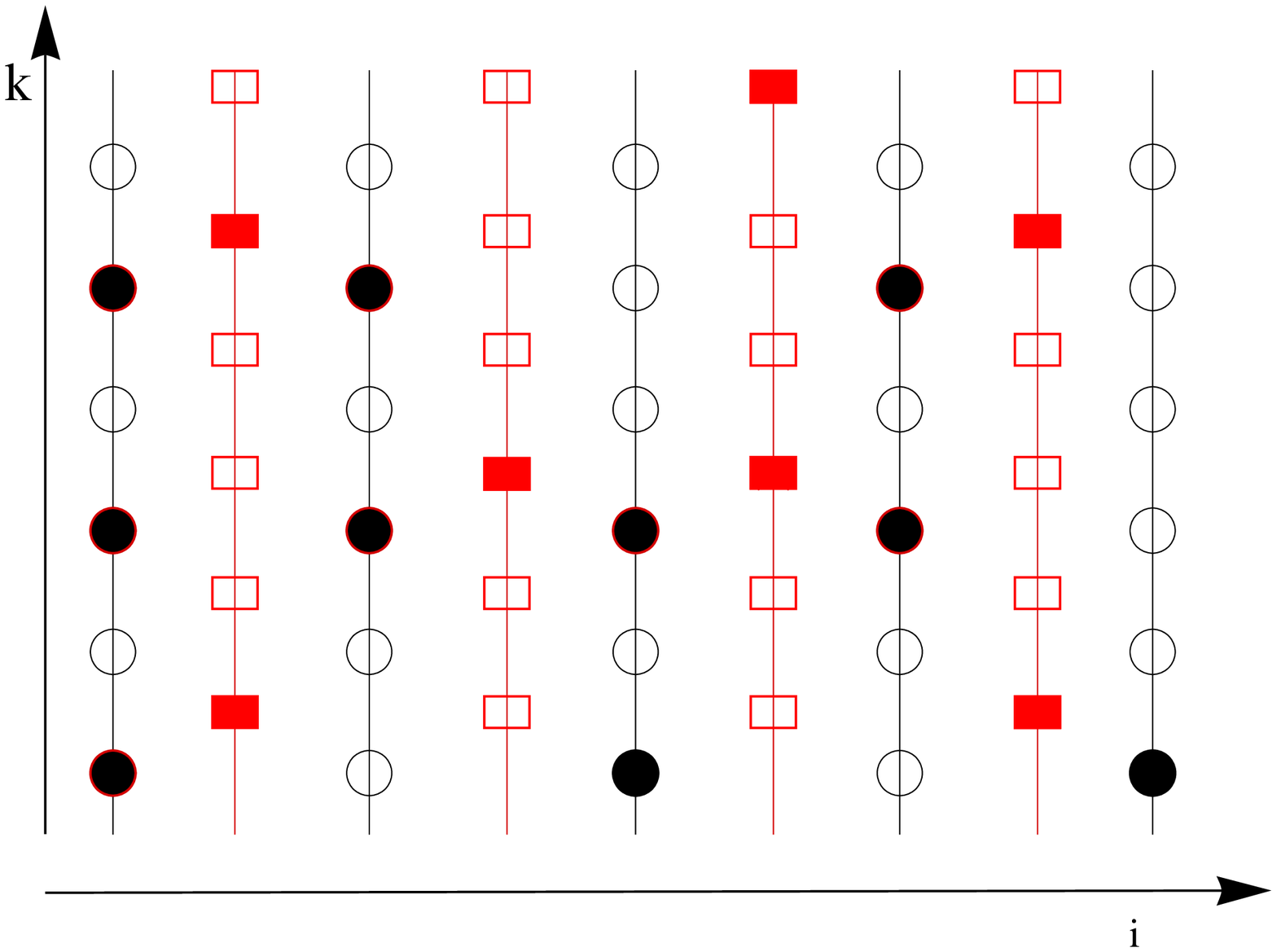} 
\caption{A configuration of particles (black disks) and scatterers (rectangles) on five rings. Periodic boundary conditions are imposed on the vertical direction.}
\end{figure}

In every other case, the particle located at site $(k,i)$ simply moves, upward in the vertical direction, to $(k+1,i)$.
On top of the dynamics induced by the dynamical system $F$ on the occupation variables, we add a stochastic update of the variables located on the boundaries of the system $B_-$ and $B_+$.  This models the coupling of the system to reservoirs of particles at chemical potentials $\rho_-$ and $\rho_+$.
We obtain the following dynamics : given $\sigma(\cdot;t-1)$, we define $\sigma(\cdot;t)$ for all $t\in\bbN^*$ by
$$
\sigma(x;t)=\left\{
\begin{array}{lll}
\sigma(F^{-1}(x);t-1)\quad {\rm if} \quad x\notin B_{-}\cup B_{+}\\
\\
\sigma^-_{x}(t-1)\quad {\rm if}  \quad x\in B_{-}\\
\\
\sigma^+_{x}(t-1)\quad {\rm if} \quad  x\in B_{+}
\end{array}
\right.
$$
The families of random variables $\{\sigma^-_{x}(t):x\in B_{-},\, t\in\bbN\}$ and $\{\sigma^+_{x}(t):  x\in B_{+},\, t\in\bbN\}$  consist of independent Bernoulli variables with respective parameters $\rho_-$ and $\rho_+$.

We define the current of particles at time $t\in\bbN$ between hyperplanes $\cC^l=\{x\in\cC: i_d=l\}$ and $\cC^{l+1}$, $l\in\{1,\ldots,N-1\}$ :
\begin{equation}
J(l,t)=\frac 1 {N^d} \sum_{(k,i)\in\cC^l} c(k,i(i+e_d))(\sigma(k,i;t)-\sigma(k,i+e_d;t)),
\label{current}
\end{equation}
$c(k,ij)$ was defined in (\ref{c_def}) and $(e_1,\ldots,e_d)$ is the canonical basis of $\bbR^d$.
\newpage
We are now ready to state our main result.

\begin{theorem}
\label{mainresult}
Let  $d\geq 7$, $\rho_I,\rho_+,\rho_-\in (0,1)$ and  $\xi$ a family of Bernoulli random variables of parameter $\mu$ and $\{\sigma(x;0): x\in \cC\}$ be a set of independent Bernoulli random variables with $\bbE[\sigma(x;0)]=\rho_-$ if $x\in B_{-}$, $\bbE[\sigma(x;0)]=\rho_+$  if $x\in B_{+}$, and $\bbE[\sigma(x;0)]=\rho_I$ if $x\notin B_{-}\cup B_{+}$.  
\begin{enumerate}
\item For any $N\in\bbN^*$ and any $t\geq \overline t=N^{d+1}$,  $J(l,t)=J(l,\overline t):=\overline J(l)$, the equality holds in law. 
\item For any $\delta>0$ and any $l\in{1,\ldots,N-1}$,
\begin{equation}
\lim_{N\to\infty}\bbP[\left|N\bar J(l  )-\kappa(\mu)(\rho_--\rho_+)\right|>\delta]=0,
\label{Fickstationary}
\end{equation}
where $\kappa(\mu)=\mu(1-\mu)^{4d-2}$.
\item There exist  random variables $\{L(l,t):l\in\{1,\ldots,N-1\},t>0\}$ such that for every $\delta>0$, $\varepsilon>0$, $t>0$ and $l\in{1,\ldots,N-1}$,
\begin{equation}
\lim_{N\to\infty}\bbP[\left|N(J(l  ,t N^2)+L(l, t N^2))-\kappa(\mu)(\rho_--\rho_+)\right|>\delta]=0.
\label{Fickmain}
\end{equation}
and $L(l,t)$  satisfies :

\begin{equation}
\bbP[|L(l, t N^2)|>\varepsilon]\leq C \varepsilon^{-1}\exp-\kappa(\mu) t+O(\frac{1}{N^{\frac{5}{2}}}),
\label{Fickbismain}
\end{equation}
Moreover, for any $l$ and  for any $t\geq N^{d+1}$, $L(l,t)=0$.
\end{enumerate}
\end{theorem} 

\noindent{\bf Remark.} Throughout the paper, we use the generic notation $C$ and $c$ for constants that depend only on the dimension of the system.  Their value may change from one line to the next.

\begin{figure}[thb]
\includegraphics[width = .89\textwidth]{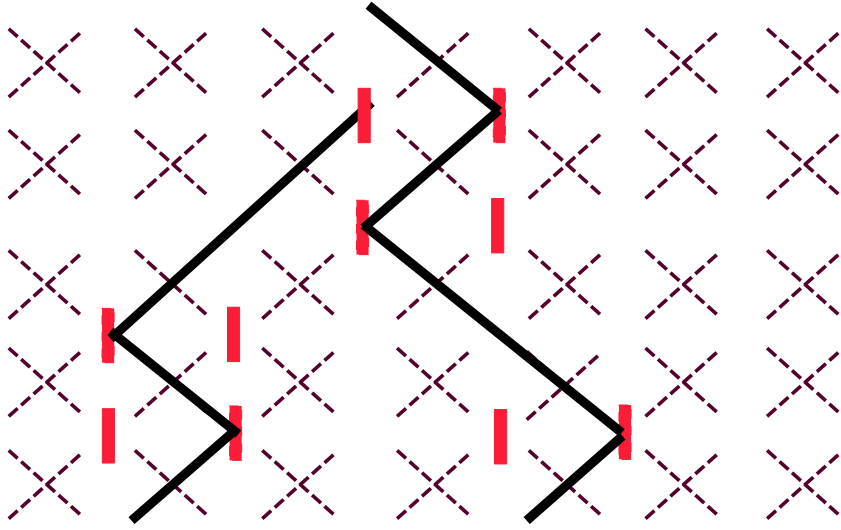} 
\caption{The dynamics is similar to the motion of particles in the mirrors model (with periodic boundary conditions in the vertical direction) in this figure.  Particles move at speed one on the edges of the dashed lattice and get reflected by the mirrors.  All particles start their motion at the vertices marked by the crosses.  Their presence at the crossed vertices is recorded every two units of time. Mirrors always appear in pairs. Pairs can overlap. The {\it presence} of a pair of mirrors corresponds to the {\it absence} of scatterers in our model.}
\end{figure}
\section{Current of particles and number of crossings}

The goal of this section is to provide a relation between the current and the number of crossing orbits from $B_-$ to  $B_+$ induced by $F$.  This relation holds for {\it fixed}  configurations of scatterers $\xi$. The only randomness that appears in the context of this section comes from the initial distribution of particles and the injection of particles at the boundaries.

We first observe that
\begin{equation}
J(l,t)=\frac 1 {N^d}\sum_{x\in\cC}\sigma(x;t)\Delta(x,l)
\label{cud}
\end{equation}
where
\begin{equation}
\Delta(x,l)={\bf 1}_{x\in\cC^l,F(x)\in\cC^{l+1}}-{\bf 1}_{x\in\cC^{l+1}, F(x)\in \cC^l}.
\label{Delta}
\end{equation}
$\Delta(x,l)$ takes the value $+1$ (resp. $-1$), if following the orbit in which it is included, $x$ crosses the slice $\cC^l\cup\cC^{l+1}$ from left to right (resp. from right to left).
For any $x\in\cC$, we define the exit time:
\begin{equation}
t_B (x)=\inf\{t>0: F^t(x)\in B_{-}\cup B_{+}\}.
\label{Exit_time}
\end{equation}
Next, we need to define {\it excursions}.  For any $x\in B$ such that $t_B(x)>1$, we define the {\it excursion} $\cE(x)=\{x,F(x),\ldots,F^{t_B(x)-1}(x)\}$. If $t_B(x)=1$, then $\cE(x)=\{x\}$. The set of excursions included in an orbit is a partition of the orbit. Depending on where the excursions start and finish we call those excursions {\it left-to-right} crossings, {\it right-to-left} crossings, {\it left-to-left} paths and {\it right-to-right} paths. An {\it internal} orbit $\cO$ is an orbit such that $\cO\cap B=\emptyset$.  The collection of all excursions and  internal orbits forms a partition of $\cC$ and there is a one-to-one correspondence between the set of points of $B_-\cup B_+$ and the set of all excursions, i.e. each point of $B_-\cup B_+$  belongs to exactly one excursion and each excursion contains at most one point of $B_-\cup B_+$. Thus (\ref{cud}) may be written :
\begin{equation}
J(l,t)=\frac 1 {N^d}\left(\sum_{\cO\cap B=\emptyset} J(l,t;\cO)+\sum_{x\in B} J(l,t,\cE(x))\right).
\label{decomporb}
\end{equation}
where
\begin{eqnarray}
J(l,t,\cO)&=&\sum_{x\in\cO}\sigma(x;t)\Delta(x,l)\\
&=& \sum_{n=0}^{T(y)}\sigma(F^n(y);t)\Delta(F^n(y),l)
\label{jorb}
\end{eqnarray}
for any $y\in \cO$ and
\begin{eqnarray}
J(l,t;\cE(x))&=&\sum_{y\in\cE(x)}\sigma(y;t)\Delta(y,l)\\
&=&\sum_{n=0}^{t_B(x)-1}\sigma(F^n(x);t)\Delta(F^n(x),l).
\label{dexccu}
\end{eqnarray}

We denote 
by $\cN_{\pm}$ the numbers of crossings from $B_{\pm}$ to $B_{\mp}$ induced by $F$, i.e. $\cN_{\pm}=|S_\pm|$ where $S_\pm$ is given by
$$
S_\pm=\{x \in B_{\pm}: F^1(x)\notin B_{\pm},\ldots, F^{s-1}(x)\notin B_{\pm}, F^s(x)\in B_{\mp} \textrm{ for some }s\in\N^*\}.
$$
One notes that $\cN_+=\cN_-$.  Indeed, since every orbit is closed, it must contain as many left-to-right than right-to-left crossings.  Thus, we set $\cN= \cN_+=\cN_-$.
  We define also the set of points of the boundaries $B_-$ and $B_+$ that are the starting points of excursions that  cross $\cC^l\times \cC^{l+1}$ at least once after time $t$ :
\begin{equation}
S_{\pm}(l,t)=\{x \in B_{\pm}:\exists s,\;t < s < t_B(x),\; \Delta(F^s(x),l)\neq 0 \}
\label{Slt}
\end{equation}
and we use the notation $\cN_{\pm}(l,t)=|S_\pm(l,t)|$. 
\begin{prop}
\label{Pcrossing}
Let  $\{\sigma(x;0): x\in \cC\}$ be a set of independent Bernoulli random variables with $\bbE[\sigma(x;0)]=\rho_\pm\in (0,1)$ if $x\in B_{\pm}$, and $\bbE[\sigma(x;0)]=\rho_I\in (0,1)$ if $x\notin B_{-}\cup B_{+}$. Then, for every $\delta>0$, every $\xi$ and every $t\in\bbN^*$,
$$
\bbP\left[\left|J(l,t)-\frac{{\mathcal N}}{N^d}(\rho_--\rho_+)+L(l,t)\right|\geq \delta \right]\leq 2\exp(-\delta^2N^d), l\in{1,\ldots,N-1}.
$$
and $L(l,t)$
satisfies
$$
| L(l,t)|\leq \frac{3}{N^d}(\cN_-(l,t)(\rho_-+\rho_I)+\cN_+(l,t)(\rho_++\rho_I)).
$$
If $t\geq N^{d+1}$, $L(l,t)=0$.
\end{prop}

\begin{proof} 
The fact that  $L(l,t)=0$, when $t\geq N^{d+1}$ follows from the definition of  $\cN_\pm(l,t)$ and the fact  that, since the dynamics is injective, the length of any orbit is bounded by the number of sites in $\cC$, namely $N^{d+1}$.
From the definition of $\sigma(\cdot;t)$ in terms of $\sigma(\cdot;t-1)$, it follows by induction that at any given time $t\in\N$, the random variables $\{\sigma(x;t): x\in \cC\}$ are independent. Therefore, $N^dJ(l,t)$ is just a sum of $2N^d$  independent Bernoulli random variables, each of which appears with a deterministic $\{-1,0,+1\}-$valued multiplier. 
Consequently, Hoeffding's concentration inequality guarantees that for all $\delta>0$, 
$$\bbP\left(\left|J(l,t)-\bbE\left[J(l,t)\right]\right|\geq \delta\right)\leq 2\exp({-\delta^2}{N^d}).$$
To conclude the proof, it remains to show that : 
\begin{equation}
\bbE\left[J(l,t)\right]=\frac{\cN}{N^d}(\rho_--\rho_+)+L(l,t).
\label{averagecurrent}
\end{equation}
Using (\ref{decomporb}), we have

\begin{equation}
\bbE[J(l,t)]=\frac 1 {N^d}\left(\sum_{\cO\cap B=\emptyset} \bbE[J(l,t;\cO)]+\sum_{x\in B} \bbE[J(l,t,\cE(x))]\right).
\label{decomporb1}
\end{equation}
And for any $y\in\cO$,
\begin{equation}
\bbE[J(l,t;\cO)]=\sum_{n=0}^{T(y)}\bbE[\sigma(F^n(y);t)]\Delta(F^n(y),l).
\end{equation}

If an orbit $\cO$ is {\it internal}, then $\bbE[J(l,t;\cO)]=0$. Indeed, for any $n$ such that $0\leq n\leq T(y)$ and $\forall t\geq 0$, $\bbE[\sigma(F^n(y);t)]=\rho_I$. Moreover, since all orbits are closed, there must be as many left-to-right than right-to-left crossings giving contribution of opposite signs and thus the sum vanishes. 
We turn now to the contribution of the second term of (\ref{decomporb1}).
For $x\in B_\pm$,
\begin{equation}
\bbE[J(l,t;\cE(x))]=\rho_\pm\sum_{n=0}^{t\wedge (t_B(x)-1)}\Delta(F^n(x),l)+\rho_I\sum_{n=t+1}^{t_B(x)-1}\Delta(F^n(x),l),
\label{exccu}
\end{equation}
with the convention that the second sum is zero if $t_B(x)\leq t+2$.  This equation follows because if $x\in B_\pm$, then for $ n\leq t\wedge (t_B(x)-1)$, $\sigma(F^n(x);t)$ is a Bernoulli random variable of parameter $\rho_\pm$  and for $n>t$, $\sigma(F^n(x);t)$ is a Bernoulli random variable of parameter $\rho_I$.  We note  that in the above sums, successive non-zero terms have opposite signs : it is impossible to cross the slice $\cC^l\cup\cC^{l+1}$ successively twice in the same direction.
Keeping this in mind,  let us examine the added contributions of each type of excursions $\cE(x)$.
\begin{itemize}

\item \textbf{left-to-left paths.} 
\noindent  In that case $t_B(x)> 1$ and both $x$ and $F^{t_B(x)}(x)$ belongs to $ B_-\cap \cO$. If $x\notin S_-(l,t)$, then the second sum in (\ref{exccu}) vanishes because it contains only zero terms. But if $x\notin S_-(l,t)$, the first one also vanishes because  it is a left-to-left path and there must be as many $+1$ terms than there are $-1$ terms.  If $x\in S_-(l,t)$, then we use the fact that both sums in (\ref{exccu}) are bounded by $1$ in absolute value, thus the total contribution of such excursions to (\ref{decomporb1}) is bounded in absolute value by $\cN_-(l,t)(\rho_-+\rho_I)$.
\item \textbf{right-to-right paths.} In that case both $x$ and $F^{t_B(x)}(x)$ belongs to 
$ B_+\cap \cO$  and $t_B(x)> 1$.  Using exactly the same arguments than for the left-to-left paths, one gets that the contribution to (\ref{decomporb1}) of the right-to-right paths is positive and bounded in absolute value by $\cN_+(l,t)(\rho_++\rho_I)$.
\item \textbf{left-to-right crossings}. In that case, $x\in B_-$ and $F^{t_B(x)}(x)\in B_+$ and thus $x$ belongs to $S_-$.  If $x\in S_-\cap S^c_-(l,t)$, then the second sum in (\ref{exccu}) vanishes because all its terms are zero.  But then, since it is a  left-to-right crossing, there must be exactly one more $+1$ term than there are $-1$ terms in the first sum.
Thus, in that case $\bbE[J(l,t;\cE(x))]=\rho_-$.  Therefore, one gets that the total contribution to (\ref{decomporb1}) of those type of crossings is equal to 
\begin{equation}
\rho_-|S_-\cap  S^c_-(l,t)|=\rho_-(\cN-|S_-\cap  S_-(l,t))|)
\label{main}
\end{equation}
and, obviously $|S_-\cap  S_-(l,t))|\leq \cN_-(l,t)$.
Next, if $x\in S_-\cap S_-(l,t)$, then we use again that both sums in (\ref{exccu}) are bounded by $1$ in absolute value and therefore the contribution of this type of excursions is bounded by $\cN_-(l,t)(\rho_-+\rho_I)$.
\item \textbf{right-to-left crossings}. This case is analogous to the previous one, except that the sign in (\ref{main}) is reversed and thus the total contribution of this type of crossings to (\ref{decomporb1}) is 
$$
-\rho_+(\cN-|S_-\cap  S_+(l,t))|).
$$

\end{itemize}
Putting all contributions together, we get (\ref{averagecurrent}) with $L(l,t)$ satisfying the bound of the proposition.
\end{proof}

\section{Lazy random walks}

\noindent 
We will see later that as long as they don't make ``loops" (to be defined later) or stay suficiently far from each other , the horizontal components of the orbits have the same law than independent lazy random walks.  This section is devoted to defining lazy random walks and studying their (self-) intersection properties in high dimension.
Let us consider  a finite box $\Lambda=\{i=(i_1,\ldots,i_d)\in\bbZ^d: \quad 1\leq i_l\leq N,\; 1\leq l\leq d\}$.
We put periodic boundary conditions on the $d-1$ first components and therefore addition on those components are to be understood modulo $N$. The distance between points in $\Lambda$ is defined as it was defined for points in $\Lambda_2$ in section 2.
The two boundaries corresponding to the sides of  the hypercube $\Lambda$ orthogonal to the $d$-th direction are :
$$
b_{-}=\{i\in\Lambda : i_d=1\}\quad{\rm and} \quad b_{+}=\{i\in\Lambda : i_d=N\}.
$$
and we define $b=b_{-}\cup b_{+}$.  We define a lazy random walk $W_t(i), t\in\bbN$ on $\Lambda$ with starting point $i\in\Lambda$ ($W_0(i)=i$) in the following way. Let $\nu\in (0,1)$ such that $\nu\leq\frac 1{2d}$, $\{W_t(i) : t\in\bbN\}$ is the Markov chain such that for $i$ and $j\in\Lambda$,
\begin{equation}
\bbP[W_{t+1}(i)=j'|W_t(i)=j]=\left\{\begin{array}{llll}
\nu & {\rm if} & d(j,j')=1\\
1-2d \nu  & {\rm if} & j=j' & j\notin b\\
1-(2d-1) \nu  & {\rm if} & j=j' & j\in b\\
0& {\rm if} & d(j,j')>1 
\label{markov1}
\end{array}
\right.
\end{equation}
We will compare this lazy random walk on $\Lambda$ to a lazy random walk on $\bbZ^d$.  The law of the latter will be denoted by $\bbP_\infty$ and it is defined in an analogous way :
\begin{equation}
\bbP_\infty[W_{t+1}(i)=j'|W_t(i)=j]=\left\{\begin{array}{llll}
\nu & {\rm if} & \|j-j'\|=1\\
1-2d \nu  & {\rm if} & j=j' \\
0& {\rm if} & \|j-j'\|>1 
\label{markov2}
\end{array}
\right.
\end{equation}
The first hitting time of the boundary $b$ starting from site $i$ is : 
\begin{eqnarray}
\tau_B(i)&=&\inf\{t\geq 0: W_t(i)\in b \}
\label{time_exit}
\end{eqnarray}
The following proposition is completely standard for simple random walks. For completeness, we give a quick proof in the case of the lazy random walks that are of interest to us.

\begin{prop}
\label{exittime}
\begin{eqnarray}
\bbE[\tau_B(i)]&\leq& C N,\quad d(i,b)=1\nonumber\\
\bbP[\tau_B(i)>t]&\leq& C \exp(-\nu\frac t {N^2}), \quad \forall i\in \Lambda
\label{exp_bound}
\end{eqnarray}
\end{prop}
\begin{proof}
We focus on the second inequality, as the first one may be derived by methods similar to the ones used below.
We bound uniformly $\bbE[\exp(\frac{\nu\tau_B(i)}{N^2})]$, which gives the result by the exponential Tchebychev inequality. We consider $\bbE[\exp(\lambda\tau_B(i))]$ with $\lambda>0 $. We need to study the law of the exit time $\tau_B(i)$, for $i\in\Lambda$. If $i\in b$, then 
$\bbP[\tau_B(i)=0]=1$.
For $i\notin b$, we have :
\begin{eqnarray}
\bbP[\tau_B(i)=n]&=&\bbP[\tau_B(i)=n|W_1(i)=i+e_d]\bbP[W_1(i)=i+e_d]\nonumber\\
&+&\bbP[\tau_B(i)=n|W_1(i)=i-e_d]\bbP[W_1(i)=i-e_d]\nonumber\\
&+&\bbP[\tau_B(i)=n|W_1(i)\neq i\pm e_d]\bbP[W_1(i)\neq i\pm e_d]
\end{eqnarray}
Setting $f(i,n)= \bbP[\tau_B(i)=n]$, we get the system of equations :
 \begin{equation}
\nu^{-1}(f(i,n)-f(i,n-1))=\left\{
\begin{array}{lll}
f(i+e_d,n-1)+f(i-e_d,n-1)-2f(i,n-1)& d(i,b)>1& n\geq 1 \\
f(i-e_d,n-1)-2 f(i,n-1)& d(i,B_{+})=1&n\geq 2 \\
f(i+e_d,n-1)-2f(i,n-1))& d(i,B_{-})=1&n\geq 2\\
\end{array}
\right.
\end{equation}
and
\begin{equation}
f(i,n)=\left\{\begin{array}{lll}
0 &i \notin b & n=0 \\
1 &i\in b & n=0\\
\nu & d(i,B)=1 & n=1.\\
\end{array}
\right.
\end{equation}
Writing $h(i,\lambda)=\bbE[e^{\lambda\tau_B(i)}]$, we have the equations :
 \begin{equation}
\nu^{-1}(e^{-\lambda}-1)h(i,\lambda)=\left\{
\begin{array}{lll}
h(i+e_d,\lambda)+h(i-e_d,\lambda)-2h(i,\lambda)& d(i,b)>1& \\
 h(i-e_d, \lambda)-2 h(i,\lambda)+1& d(i,B_{+})=1& \\
 h(i+e_d,\lambda)-2 h(i,\lambda)+1& d(i,B_{-})=1.& \\
\end{array}
\right.
\end{equation}
The explicit solution of this system is :
\begin{equation}
h(i,\lambda)=\frac{\cos(\omega(\lambda)(i_d-1))+\cos(\omega(\lambda)(i_d-N))}{1+\cos(\omega(\lambda)(N-1))}
\label{gen_sol}
\end{equation}
where $\omega(\lambda)$ is the solution of $\cos \omega(\lambda)=1-\frac{1}{2\nu}(1-e^{-\lambda})$ for $\lambda$ small enough.
Setting $\lambda=\frac{\nu}{N^2}$ and remembering that $\arccos(1-x)=(2x)^\frac{1}{2}+O(x^\frac{3}{2})$, one concludes easily.
\end{proof}
\noindent Let $m\in\bbN^*$, then,  the smallest time $\tau_L(i,m)$ at which the lazy random walk starting at $i\in\bbZ^d$ comes back to the neighbourhood of a point it has visited at a multiple of $m\in\bbN$ steps back in time is defined by :
\begin{equation}
\tau_L(i,m)=\inf\{t\in\bbN^*,\;\exists \,q\in\bbN^*, d(W_{t-qm}(i),W_{t}(i))|\leq 3\}.
\label{time_loop}
\end{equation}

\begin{prop}
\label{loop_exit}
{\it If $d\geq 7$, for any $i\in \Lambda$ and any $N\in\bbN^*$}
\begin{eqnarray}
\bbP[\tau_L(i,N)& \leq& \tau_B(i)]\leq C\frac{\bbE[\tau_B(i)]}{N^{\frac d 2}}+o(\frac 1{N^{\frac{d}{2}-1}})
\end{eqnarray}
\end{prop}

\begin{proof}
The proposition being trivial if $i\in b$, we assume below $i\notin b$.
We introduce
\begin{equation}
\tau_L(i,m,\zeta)=\inf\{t\in\bbN^*,\;\exists q\in\bbN^*, W_{t-qm}(i)=W_{t}(i)+\zeta\}.
\end{equation}
and
\begin{equation}
q_L(i,m,\zeta)=\inf\{q\in\bbN^*: W_{\tau_L(i,m)-qm}(i)=W_{\tau_L(i,m)}(i)+\zeta\} .
\end{equation}
We observe that
$$
\bbP[\tau_L(i,m) \leq \tau_B(i)]\leq \sum_{\zeta\in\{-3,\ldots,3\}}\bbP[\tau_L(i,m,\zeta) \leq \tau_B(i)].
$$
and we prove the bound of the proposition for $\tau_L(i,m,0)$, which we denote by $\tau_L$.  Other values of $\zeta$ are treated in an analogous way.  We also use the notation $q_L=q_L(i,m,0)$.

We compute :
\begin{eqnarray}
\bbP[\tau_L \leq \tau_B(i) ]&\leq& \sum_{n=1}^{A_N}\bbP[\tau_L \leq \tau_B(i),\tau_B(i)=n]+\bbP[\tau_B(i)>A_N]\nonumber\\
\end{eqnarray}
where $A_N=\lfloor c_1N^2\log N\rfloor$. With this choice of $A_N$, Proposition \ref{exittime} implies that the second term may be made smaller than $C/N^\beta$ with $\beta$ as large as necessary by choosing $c_1$ sufficienty large.  Thus we focus on the first term $S_N=\sum_{n=1}^{A_N}\bbP[\tau_L \leq \tau_B(i),\tau_B(i)=n]$.
\begin{eqnarray}
S_N&\leq&\sum_{n=1}^{A_N}\sum_{t=1}^{n-1}\sum_{q=1}^{\lfloor t/N\rfloor}\sum_{j\in\Lambda}\bbP[W_{t-qN}(i)=j,q_L=q,\tau_L=t,\tau_B(i)=n]\nonumber\\
&\leq&\sum_{n=1}^{A_N}\sum_{t=1}^{n-1}\sum_{q=1}^{\lfloor t/N\rfloor}\sum_{j\in\Lambda}\sum_{\substack{k\in\bbZ^d\\k_d=0}} \bbP_\infty[W_{t-qN}(i)=j,W_t(i)=j+kN,q_L=q,\tau_L=t,\tau_B(i)=n]\nonumber\\
&\leq&\sum_{n=1}^{A_N}\sum_{t=1}^{n-1}\sum_{q=1}^{\lfloor t/N\rfloor}\sum_{j\in\Lambda}\sum_{\substack{k\in\bbZ^d\\k_d=0}} \bbP_\infty[\tau_B(j+kN)=n-t]\bbP_\infty[W_{t}(i)=j+kN,W_{t-qN}(i)=j]\nonumber\\
\label{bou}
\end{eqnarray}
The second inequality follows from the periodic boundary conditions on the $d-1$-first components.
In the last line we have used conditioned the probability of the event of $\{\tau_B(i)=n\}$ on the event :
$$
\{W_{t-qN}(i)=j, W_t(i)=j+kN,\tau_L=t,q_L=q\}\subset \{W_{t}(i)=j+kN,W_{t-qN}(i)=j\},
$$
and next the  Markov property of the random walk.
But we also have :
\begin{eqnarray}
\bbP_\infty[W_{t}(i)=j+kN, W_{t-qN}(j)=j]&=&\bbP_\infty[W_{t-qN}(i)=j]\bbP_\infty[W_{qN}(j)=j+kN]\nonumber\\
&=&\bbP_\infty[W_{t-qN}(i)=j]\bbP_\infty[W_{qN}({\bf 0})=kN]\nonumber\\
\end{eqnarray}
where ${\bf 0}=(0,\ldots,0)$ by the Markov property and translation invariance of $\bbP_\infty$.  Moreover,
$\bbP_\infty[\tau_B(j+kN)=n-t]=\bbP_\infty[\tau_B(j)=n-t]$, if $k_d=0$.
Therefore :
\begin{equation}
S_N\leq\sum_{n=1}^{A_N}\sum_{t=1}^{n-1}\sum_{q=1}^{\lfloor t/N\rfloor}\sum_{j\in\Lambda}\sum_{\substack{k\in\bbZ^d\\ k_d=0}} \bbP_\infty[\tau_B(j)=n-t]\bbP_\infty[W_{t-qN}(i)=j]\bbP_\infty[W_{qN}({\bf 0})=kN].
\label{bounds}
\end{equation}
Using again the Markov property, we write for $qN\leq t\leq n$ :
\begin{eqnarray}
\sum_{j\in\Lambda}\bbP_\infty[\tau_B(j)=n-t]\bbP_\infty[W_{t-qN}(i)=j]&=&\sum_{j\in\Lambda}\bbP_\infty[\tau_B(i)=n-qN|W_{t-qN}(i)=j]\bbP_\infty[W_{t-qN}(i)=j]\nonumber\\
&=&\bbP_\infty[\tau_B(i)=n-qN].
\end{eqnarray}
Thus,
\begin{eqnarray}
S_N&\leq&\sum_{n=1}^{A_N}\sum_{q=1}^{\lfloor n/N \rfloor}\sum_{t=qN+1}^{n-1}\sum_{k\in\bbZ^d}\bbP_\infty[\tau_B(i)=n-qN] \bbP_\infty[W_{qN}({\bf 0})=kN]\nonumber\\
&\leq&\sum_{n=1}^{A_N}\sum_{q=1}^{\lfloor n/N\rfloor}\sum_{k\in\bbZ^d}(n-qN)\bbP_\infty[\tau_B(i)=n-qN] \bbP_\infty[W_{qN}({\bf 0})=kN]\nonumber\\
&\leq&\sum_{q=1}^{\lfloor A_N/N\rfloor}\sum_{n=qN}^{A_N}\sum_{k\in\bbZ^d}(n-qN)\bbP_\infty[\tau_B(i)=n-qN] \bbP_\infty[W_{qN}({\bf 0})=kN]\nonumber\\
&\leq&\bbE[\tau_B]\sum_{q=1}^{\lfloor A_N/N\rfloor }\sum_{k\in\bbZ^d} \bbP_\infty[W_{qN}({\bf 0})=kN]\nonumber\\
&\leq&\bbE[\tau_B]\sum_{q=1}^{\lfloor A_N/N\rfloor }\left (\bbP_\infty[W_{qN}({\bf 0})={\bf 0}]+\sum_{1<\|k\|\leq B_N} \bbP_\infty[W_{qN}({\bf 0})=kN]+\sum_{B_N<\|k\|\leq q} \bbP_\infty[W_{qN}({\bf 0})=kN]\right)\nonumber\\
\label{loop3}
\end{eqnarray}
where $B_N=c_2(\frac q N)^{\frac 1 2}(\log q N)^{\frac 1 2}$.  Proposition 2.4.4 in \cite{Lawler} guarantees that 
$$
\bbP_\infty[W_{qN}({\bf 0})={\bf 0}]\leq \frac C{(qN)^{d/2}}
$$
so that 
$$
\sum_{q=1}^{\lfloor A_N/N\rfloor }\bbP_\infty[W_{qN}({\bf 0})={\bf 0}]\leq \sum_{q=1}^{\lfloor A_N/N\rfloor }\frac C{(qN)^{d/2}}\leq\frac C{N^{d/2}}
$$
for $d>2$.
The proposition will be proven once we show that the contribution of the two remaining terms in (\ref{loop3}) is $o(\frac{1}{N^{d/2}})$.
Theorem 4.3.1 in \cite{Lawler} implies that
$$
\bbP_\infty[W_{qN}({\bf 0})=kN]\leq\frac C {\|kN\|^{d-2}}
$$
so that
\begin{eqnarray}
\sum_{q=1}^{\lfloor A_N/N\rfloor}\sum_{1<\|k\|\leq B_N} \bbP_\infty[W_{qN}({\bf 0})=kN]&\leq&\sum_{q=1}^{\lfloor A_N/N\rfloor}\sum_{1<\|k\|\leq B_N}\frac C {\|kN\|^{d-2}}\nonumber\\
&\leq& C \frac{1}{N^{d-2}}\sum_{q=1}^{\lfloor A_N/N\rfloor}\sum_{n=2}^{B_N} n\nonumber\\
&\leq& C \frac{1}{N^{d-3}}(\log (N\log N))^3
\end{eqnarray}
which is $o(\frac 1 {N^{\frac d 2}})$ for $d\geq 7$.  
Let us look now at the third term of (\ref{loop3}) :
\begin{eqnarray}
\sum_{q=1}^{\lfloor A_N/N\rfloor}\sum_{B_N<\|k\|\leq q} \bbP_\infty[W_{qN}({\bf 0})=kN]&\leq& \sum_{q=1}^{\lfloor A_N/N\rfloor} \bbP_\infty[\|W_{qN}({\bf 0})\|>c_2(qN\log (qN))^\frac{1}{2}]\nonumber\\
&\leq&\sum_{q=1}^{\lfloor A_N/N\rfloor} d\; \bbP_\infty[|\hat W_{qN}(0)|>\frac{c_2}{d}(qN\log (qN))^\frac{1}{2}] \nonumber\\
&\leq & \sum_{q=1}^{\lfloor A_N/N\rfloor }\frac{1}{(qN)^\alpha}
\end{eqnarray}
where $\hat W_t(0)$ is a 1D lazy random walk starting at $0$.  The last inequality is obtained by using Hoeffding's inequality. $\alpha$ can be made as large as necessary by choosing $c_2$ sufficiently large in the definition of $B_N$. Thus the last sum above may be made $o(\frac{1}{N^{\frac  d 2-1}})$ and this concludes the proof.
\end{proof}

Let two independent random walks $\{W_t(i): t\in\bbN\}$ and $\{W_t(i'):t\in\bbN\}$ starting from two points  $i,i'\in\Lambda$ and integers $m\in\bbN$, define the ``collision" times :
$$
\tau_I(i\to i',m)=\inf\{t>0: \exists q\in\bbN,\; d(W_t(i),W_{t-qm}(i'))\leq 3\},
$$
and
\begin{equation}
\tau_I(i,i',m)=\tau_I(i\to i',m)\wedge\tau_I(i'\to i,m).
\label{intersection_lazy}
\end{equation}
We recall the result of Erd\"os and Taylor that is of interest to us, Lemma 9 in \cite{Erdos} :
\begin{lemma}
\label{Erdos}
Let $\{S_t(i) : t\in\bbN\}$ and $\{S_t(i') : t\in\bbN\}$ two independent symmetric random walks on $\bbZ^d$, with $d>4$ and with starting points   $i,i'\in \Lambda$ such that $\rho=\|i-i'\|>0$, then
\begin{equation}
\bbP[\{S_t(i) : t\in\bbN\}\cap \{S_t(i') : t\in\bbN\}\neq \emptyset]\leq \frac C {\rho^{d-4}}.
\label{mainsta}
\end{equation}
 \end{lemma}
We define $W(i,[0,t])=\{j\in\bbZ^d: \exists s\in[0,t], W_s(i)=j\}$ and for any $A\subset\Lambda$, $\overline A=\{j\in\Lambda: \exists i\in A, d(j,i)\leq 3\}$. The following lemma will be also helpful.
 \begin{lemma}
Let $W_t(i)$ a symmetric lazy random walk starting at $i\in\bbZ^d$ and $\lambda>0$ large enough, for $t=c_1\lambda^2\log\lambda$ and $i\in\bbZ^d$ and $j\in\bbZ^d$ such that $\|i-j\|>c_2\lambda \log\lambda$, then
 $$
 \bbP[\overline {W(i,[0,t])}\cap \overline {W(j,[0,t])}\neq \emptyset]\leq \frac {C} {\lambda^{\alpha}}
 $$
 where $\alpha=\frac{c_2^2}{18d^2 c_1}$.
 \label{finitet}
 \end{lemma}
 \begin{proof}
 Since $\|i-j\|>c_2\lambda\log\lambda$, then there is at least one $l$ such that $1\leq l\leq d$ and $|i_l-j_l|>c_2\lambda\log\lambda/d$.  Moreover
\begin{eqnarray}
\bbP[\overline{W(i,[0,t])}\cap \overline{W(j,[0,t])}\neq \emptyset]&\leq& \bbP[\overline{W^l(i,[0,t])}\cap \overline{W^l(j,[0,t])}\neq \emptyset].\nonumber\\
 \end{eqnarray}
 $W^l(i)$ is the $l$-th component of the lazy random walk and has itself the law of a 1D lazy random walk starting at $i_l$. Therefore, it is enough to consider the case of 1D lazy random walks on $\bbZ$.
 
 \begin{eqnarray}
 \bbP[\overline{W(i,[0,t])}\cap \overline{W(j,[0,t])}\neq \emptyset]&\leq& \bbP[W(i,[0,t])\cap B(i,\frac{|i-j|}{3})\neq\emptyset]\nonumber\\
&&+\bbP[W(j,[0,t])\cap B(j,\frac{|i-j|}{3})\neq\emptyset]\nonumber\\
&\leq& 2\,\bbP[|W(0,[0,t])|>\frac{|i-j|}{3}]\nonumber\\
&\leq& 2\bbP[\max_{s\in[0,t]} |W_s(0)|>\frac{|i-j|}{3}]\nonumber\\
&\leq & 4 \exp\left(-\frac{|i-j|^2}{18t}\right),
 \end{eqnarray}
 from which (\ref{mainsta}) follows.
 The first inequality holds for any $i,j$ such that $|i-j|>18$ which is guaranteed by the hypothesis of the lemma by choosing $\lambda$ large enough.
The last inequality is justified by the following standard application of Doob's maximal inequality to the non-negative sub-martingale $\{e^{\theta W_t(0)}\}_{t\geq 0}$ $(\theta\geq 0)$. For any $u>0$,
 \begin{eqnarray*}
\bbP[\max_{s\in[0,t]} W_s(0)>u] & = & \bbP[\max_{s\in[0,t]} e^{\theta W_s(0)}>e^{\theta u}]\\
& \leq & e^{-\theta u}\bbE[e^{\theta W_t(0)}]\\
& \leq & e^{-\theta u}(\cosh \theta)^t\\
& \leq & e^{-\theta u+\frac{\theta^2}{2}t}.
 \end{eqnarray*}
The second inequality follows from the convexity of the map $x\to e^{\theta x}$ which implies
\begin{eqnarray*}
\bbE[e^{\theta X}]&\leq& \cosh \theta +\bbE[X]\sinh\theta\\
&\leq &\cosh\theta
\end{eqnarray*}
for any random variable such that $|X|\leq 1$ and $\bbE[X]=0$.

Choosing $\theta=\frac{u}{t}$ yields 
\begin{eqnarray*}
\bbP[\max_{s\in[0,t]} W_s(0)>u] & \leq & e^{-\frac{u^2}{2t}},
 \end{eqnarray*}
 and by symmetry,  
 \begin{eqnarray*}
\bbP[\max_{s\in[0,t]} |W_s(0)|>u] & \leq & 2e^{-\frac{u^2}{2t}}.
 \end{eqnarray*}
 \end{proof}
We now come back to the lazy random walk defined on $\Lambda$.
\begin{prop}For $i$ and  $i'$ in $\Lambda$ such that $d(i,i')>N^{\frac 3 4}$  and $d\geq 7$,
$$
\bbP[\tau_I(i,i',m)< \tau_B(i)\vee\tau_B(i') ]\leq \frac{C}{N^{\frac{9}{4}}}, \quad \forall m\in\bbN.
$$
\label{intersection}
\end{prop}
\begin{proof}
In this proof, we use the notation $\tau_I=\tau_I(i,i',m)$ and $\tau_B=\tau_B(i)\vee\tau_B(i')$ and start by decomposing :
\begin{eqnarray}
\bbP[\tau_I< \tau_B ]&=&\bbP[\tau_I< \tau_B,\tau_B\leq c N^2\log N ]+\bbP[\tau_I< \tau_B,\tau_B> c N^2\log N ].
\end{eqnarray}
We notice that by (\ref{exp_bound}), the second term may be made smaller than $C/N^\alpha$ with $\alpha$ as large as necessary by chosing $c$ sufficienty large.
We define the event :
$$
I(i,j;t)=\{\overline {W(i,[0,t])}\cap \overline {W(j,[0,t])}\neq \emptyset\}.
$$
The first term may be bounded in the following way (we use again the notation $A_N=\lfloor cN^2\log N\rfloor$) :
\begin{eqnarray}
\bbP[\tau_I< \tau_B,\tau_B\leq c N^2\log N ]
&\leq &\bbP_\infty[\bigcup_{t=1}^{\tau_B}\bigcup_{t'=1}^{\tau_B}\bigcup_{\substack{ k\in\bbZ^d\\ k_d=0}}\{d(W_t(i),W_{t'}(i')+kN)\leq 2\},\tau_B\leq A_N]\nonumber\\
&\leq &\sum_{k\in\bbZ^d}\bbP_\infty[I(i,i'+kN;A_N)]
\end{eqnarray}
We split the RHS and get :
\begin{eqnarray}
\bbP[\tau_I< \tau_B,\tau_B\leq c N^2\log N ]&\leq &\sum_{\|k\|\leq 2}\bbP_\infty[I(i,i'+kN;A_N)] \nonumber\\
&&+\sum_{3\leq \|k\|\leq B_N}\bbP_\infty[I(i,i'+kN;A_N)] \nonumber\\
&&+\sum_{ \|k\|> B_N}\bbP_\infty[I(i,i'+kN;A_N)] \nonumber\\
\label{split}
\end{eqnarray}
and we choose $B_N=c_2\log N$. We now look at the first term of (\ref{split}).
Since by hypothesis $d(i,i')>N^{\frac 3 4}$, we have $\|i-i'-kN\|>N^{\frac 3 4}$, for any $k\in\bbZ^d$ and thus using a trivial adaptation of Proposition \ref{Erdos} to the case of lazy random walks, we obtain :
$$
\sum_{\|k\|\leq 2}\bbP_\infty[I(i,i'+kN;A_N)] \leq  \frac {C}{N^{\frac {9} 4}}.
$$
if $d\geq 7$.

We next consider the second term of (\ref{split}).  With the help of Proposition \ref{Erdos},

\begin{eqnarray}
\sum_{3\leq \|k\|\leq B_N}\bbP_\infty[I(i,i'+kN;A_N)] &\leq &\sum_{3\leq \|k\|\leq B_N}\frac{1}{\|i-i'-kN\|^{d-4}}\nonumber\\
\end{eqnarray}
But since $\|i-i'\|\leq 2N$ and $\|k\|\geq 3$, we have 
$$
\|i-i'-kN\|^{d-4}\geq c \|k\|^{d-4} N^{d-4}.
$$
Thus,
\begin{eqnarray}
\sum_{3\leq \|k\|\leq B_N}\bbP_\infty[I(i,i'+kN;A_N)] &\leq &C\sum_{3\leq n\leq B_N}\frac{n^{d-1}}{n^{d-4} N^{d-4}}\nonumber\\
&\leq & C\frac{(\log N)^4}{N^{d-4}}=o(\frac 1{N^{\frac 9 4}}).
\end{eqnarray}
when $d\geq 7$.

We show finally that the last term may also be made $o(\frac 1 {N^{\frac {9} 4}})$.
We note that $\forall i,i'\in\Lambda$  and $t,t'\leq A_N$, if $\|k\|> 5A_N/N$ then
$$
\bbP_\infty[\|W_t(i)-W_{t'}(i')-kN\|\leq 6]=0.
$$
Indeed, if $k\in\bbZ^d$ is such that $\|W_t(i)-W_{t'}(i')-kN\|\leq 6$ we necessarily have the following inequality (remember that, under $\bbP_\infty$, random walks are now defined on $\bbZ^d$) :
$$
\|k\|N\leq \|W_t(i)\|+\|W_{t'}(i')\|+6\leq t+\|i\|+t'+\|i'\|+6\leq 5A_N.
$$
Therefore,
\begin{eqnarray}
\sum_{ \|k\|> B_N}\bbP_\infty[I(i,i'+kN;A_N)]& \leq& \sum_{ B_N<\|k\|\leq 5 \frac{A_N}{N}}\bbP_\infty[I(i,i'+kN;A_N)] \nonumber\\
&\leq& (\frac{A_N}{N})^d\frac{C}{N^\alpha}\nonumber\\
\end{eqnarray}
where the last line follows from Lemma \ref{finitet} with $\lambda=N$ and because for $\|k\|>B_N$ and $N$ large enough, we have $\|i-i'-kN\|>c N\log N$.  This can be made smaller than $C/N^\beta$ with $\beta$ as large as necessary by choosing $c_2$ as large as necessary.

\end{proof}

\section{Recurrence and intersection of orbits and connection with lazy random walks}
In this section, we consider the case when $\xi$ is a collection of Bernoulli random variables of parameter $\mu$.   We show that up until the time they make ``loops" or intersect each other, the horizontal components of the orbits have the same law than a lazy random walk.
For any $x=(k,i)\in{\cC}$, we define $h(k,i)=i$ and $v(k,i)=k$.  For any $t\geq 0$, we define $H_t(x)=h(F^t(x))$ and $V_t(x)=v(F^t(x))$ that describes the motion of the horizontal component of the orbit starting at $x$. 
We define now the first time at which the orbit $(F^t(x))_{t\in\bbN}$ comes back in a neighbourhood of a point it has already visited in the past and makes a ``loop" :
\begin{equation}
 t_{L} (x)=\inf\{t>0:\exists s<t,\; V_t(x)=V_s(x'),\;d(H_t(x),H_s(x))\leq 3\}.
\label{Loop_time}
\end{equation}

An orbit starting from $x$ keeps discovering a fresh random scenery of $\xi$ variables until time $t_L(x)$.
Next, we  consider orbits with different starting points $x=(k,i)$ and $x'=(k',i')$. We define the first time when one of the two orbits visit a neighbourhood of a site that has been visited by the other orbit at some time in the past :
\begin{equation}
t_I(x\to x')=\inf\left\{t>0:\exists s<t,\;
V_t(x)=V_s(x'),\;{\rm and}\;d(H_t(x),H_s(x'))\leq 3\right\}
\end{equation}
And we define $t_I(x,x')=t_I(x\to x')\wedge t_I(x'\to x)$. The motions of two orbits starting respectively at $x$ and $x'$ are independent until time $t_I(x,x')$.
By construction of the dynamics, the vertical coordinate of a given trajectory simply moves one step ahead on a ring and thus we get the bound :
\begin{equation}
t_I(x,x')\geq (k'-k)\wedge (N-k'+k).
\label{bound_intersection}
\end{equation}

\begin{prop}
\label{connection}
Let $\xi$ a set of Bernoulli variables of parameter $\mu\in]0,1[$, 
\begin{enumerate}
\item For any $x\in\cC$,  

$$
\{H_s(x): 0\leq s\leq t_L(x)\}
$$  
has the same law than 
$$
\{W_s(h(x)): 0\leq s\leq \tau_L(h(x))\}
$$ 
where $\tau(i):=\tau(i,N)$ and$\{W_s(i) : s\in\bbN\}$ is the lazy symmetric random walk defined in the previous section with parameters $\nu=\mu(1-\mu)^{4d-2}$.

\item Let $x=(k,i)$ and $x'=(k',i')$ such that $x\neq x'$, then the collections 
$$
\{H_s(x):0\leq s\leq t_I(x,x')\}
$$ and 
$$
\{H_s(x'):0\leq s\leq t_I(x,x')\}
$$ 
are mutually independent.
\end{enumerate}
\end{prop}
\begin{proof}
From the definition of the dynamics (\ref{tau_def}) we see that for any $x=(k,i)\in\cC$, $s\geq 0$, 
$$
H_{s+1}(x)=H_s(x)+f(\xi,s,H_s(x)),
$$
where $f(\cdot,s,j)$ is a random variable measurable with respect to 
$$
\{\xi(k+s\mod N,j j'):d(j,j)'\leq 2\}.
$$
Thus, we see from the definition (\ref{Loop_time}) that the random variables in the sequence 
$$
\{f(\cdot,s,H_s(x)):0\leq s\leq t_L(x)\}
$$ 
are each measurable with respect  to independent variables and are therefore independent.  Thus $\{H_s(x): 0\leq s\leq t_L(x)\}$ coincides with the $t_L(x)$ first steps of a Markov chain. From (\ref{tau_def}), a little computation shows that for any $s\geq 1$ the law of the variables $\{f(\cdot,s,H_s(x)):0\leq s\leq t_L(x)\}$  is given by 
$$
\bbP[f(\xi,s,j)=j'-j]=\bbP[W_{s+1}(x)=j'|W_{s}(x)=j]
$$
and the statement (1) of the proposition follows.  The second part of the proposition follows by analogous arguments.
\end{proof}
\section{Convergence of the average number of crossing orbits}
We note two useful formulas for the expectation and the variance of the number of crossings from one side of the volume to the other. As there is no ambiguity here,  we use the notation $S$ for $S_-$, that was defined in section 3.

\begin{equation}
\bbE[\frac{\cN}{N^d}]=\frac 1 {N^d}\sum_{x\in B_{-}}\bbE[{\bf 1}_{x\in S}]=\bbP[(1,\ldots,1)\in S]
\label{expectproba}
\end{equation}
by rotational invariance.
\begin{eqnarray}
{\rm Var }[\frac{\cN}{N^d}]&=&\frac{1}{N^{2d}}\sum_{x,x'\in B_{-}}\bbE[{\bf 1}_{x \in S}{\bf 1}_{x'\in S}]-\bbE[{\bf 1}_{x \in S}]\bbE[{\bf 1}_{x'\in S}]\nonumber\\
&=&\frac{1}{N^{2d}}\sum_{x,x'\in B_{-}}\bbP[x \in S,x'\in S]-\bbP[x\in S]\bbP[x'\in S]\nonumber\\
&=&\frac{1}{N^d}\sum_{x\in B_{-}}\bbP[(1,\ldots, 1)\in S,x\in S]-\bbP[(1,\ldots,1)\in S]\bbP[x\in S]
\nonumber\\
\label{variance}
\end{eqnarray}
\begin{prop}
\label{crossingfick}
Let $d\geq 7$ and $\mu\in (0,1)$, then $\forall \epsilon>0$,
$$
\bbP[|\frac{\cN}{N^{d-1}}-\kappa(\mu)|>\epsilon]\leq \frac{C}{\epsilon^2N^{\frac 1 4}}
,\quad \kappa(\mu)=\mu(1-\mu)^{4d-2}$$
\end{prop}
\begin{proof}
We first prove that 
\begin{equation}
\bbE[\frac{\cN}{N^d}]=\frac {\kappa(\mu)} {N-1} +O(\frac 1 {N^{\frac 5 2}}).
\label{fick7d}
\end{equation}
We start with (\ref{expectproba}) and  use the notation ${\bf 1}$ for the point of $\Lambda$ given by ${\bf 1}=(1,\ldots,1)$.
Then, with the shorthand notation $t_{B}=t_{B}(2,\bo+e_d)$ and $ t_{L}=t_{L}(2,\bo+e_d)$,
\begin{equation}
\bbP[(1,\bo)\in S]=\bbP[(1,\bo)\in S,t_{B}< t_{L}]+\bbP[(1,\bo)\in S, t_{B}\geq t_{L}]
\label{decomposition2}
\end{equation}
For the second term we use the correspondence between the orbits and lazy random walks stated in Proposition \ref{connection}  :
\begin{equation}
\bbP[(1,1)\in S, t_{B}\geq t_{L}]\leq \bbP[t_B\geq t_L]=  \hat \bbP[\tau_B\geq \tau_L] .
\label{coupling2}
\end{equation}
$\hat \bbP$ refers to the law of the lazy random walk in Proposition \ref{connection} and we have used the notation $\tau_B=\tau_B(\bo+e_d)$, $\tau_L=\tau_L(\bo+e_d)$.
On the other hand we have :
\begin{equation}
\bbP[(1,\bo)\in S,t_{B}< t_{L}]=\hat\bbP[C\cap\{\tau_B< \tau_L\}]=\hat\bbP[C]-\hat\bbP[C\cap\{\tau_B\geq \tau_L\}]
\label{finite_gambler2}
\end{equation}
with 
$$
C=\{W_t(\bo): W_1(1)= \bo+e_d,\;\exists t>1,\;\forall s< t,\; W_s(\bo)\notin b_-,\; W_t(\bo)\in b_+\}.
$$
Because $\bbP[W_1(\bo)=\bo+e_d]=\kappa(\mu)$, the gambler's ruin argument applied to the lazy random walk gives :
\begin{equation}
\hat\bbP[C]=\frac {\kappa(\mu)} {N-1}.
\label{gambler2}
\end{equation}
Indeed, the $d$-th component $W^d_t(\bo)$ is itself a lazy random walk. 
So we get for (\ref{decomposition2}) :
\begin{equation}
\bbP[(1,\bo)\in S]=\frac {\kappa(\mu)} {N-1}+{\mathcal R}
\label{crossing_proba2}
\end{equation}
where 
\begin{equation}
|{\mathcal R}|\leq 2\; \hat\bbP[\tau_B\geq \tau_L]\leq \frac{C}{N^{d/2-1}},
\label{remainder}
\end{equation}
the last inequality follows from proposition \ref{loop_exit} and proposition \ref{exittime}. Thus, (\ref{crossing_proba2})  gives (\ref{fick7d}).
We turn now to the variance.  Defining 
$$
\Delta_x=\bbP[(1,{\bf 1})\in S,x\in S]-\bbP[(1,{\bf 1})\in S]\bbP[x\in S]
$$
we write :
\begin{eqnarray}
N^2{\rm Var }[\frac{\cN}{N^d}]&=&\frac {1} {N^{d-2}}\sum_{x\in B_{-}}\Delta_x\nonumber\\
&=&\frac {1} {N^{d-2}}(\sum_{x\in I_1}\Delta_x+\sum_{x\in I_2}\Delta_x)
\label{var1}
\end{eqnarray}
where 
$$
I_1=\{x=(k,i)\in B_-:d(i,{\bf 1})> N^\frac 3 4\}
$$
and
$$
I_2=\{x=(k,i)\in B_-:d(i,{\bf 1})\leq N^\frac 3 4\}.
$$

By using (\ref{crossing_proba2}) and (\ref{remainder}), we see that we have the a priori bound on $\Delta_x$:
\begin{equation}
\Delta_x\leq \frac C N
\label{apriori}
\end{equation}
and since there are $CN^{\frac 1 4(3d+1)}$ points in $I_2$, we get :
\begin{equation}
\frac {1} {N^{d-2}}\sum_{x\in I_2}\Delta_x\leq \frac C {N^{\frac 1 2}}
\label{var2}
\end{equation}
For the sum over points in $I_1$ we proceed differently and use the fact that random walks starting sufficienty far away from each other have a small probability of  meeting before the two of them exits, as expressed by Proposition \ref{intersection}.
To implement this strategy, we will use a decomposition with respect to the events :
 \begin{eqnarray}
A_{1,x}&=&\{t_L(1,{\bf 1})>t_B(1,{\bf 1}),t_L(x)>t_B(x), t_I((1,{\bf 1}),x)>(t_B(1,{\bf 1})\vee t_B(x))\}\nonumber\\
A_1&=&\{t_L(1,{\bf 1})>t_B(1,{\bf 1})\}\nonumber\\
A_x &=&\{t_L(x)>t_B(x)\}\nonumber
\end{eqnarray}
Namely we write ,
\begin{eqnarray}
\bbP[(1,{\bf 1})\in S,x\in S]&=&\bbP[(1,{\bf 1})\in S,x\in S, A_{1,x}]+\bbP[(1,{\bf 1})\in S,x\in S, A^c_{1,x}]\nonumber\\
\bbP[(1,{\bf 1})\in S]&=&\bbP[(1,{\bf 1})\in S, A_{1}]+\bbP[(1,{\bf 1})\in S, A^c_{1}]\nonumber\\
\bbP[x\in S]&=&\bbP[x\in S, A_{x}]+\bbP[x\in S, A^c_{x}]\nonumber\\
\end{eqnarray}
and since we have (with $x=(k,i)$) :
\begin{eqnarray}
\bbP[(1,{\bf 1})\in S,x\in S, A_{1,x}]&=&\hat\bbP[{\bf 1}\in C,i\in C, \hat A_{1,i}]\nonumber\\
&\leq&\hat\bbP[{\bf 1}\in C,i\in C, \hat A_{1},\hat A_{i}]\nonumber\\
&\leq&\hat\bbP[{\bf 1}\in C, \hat A_{1}]\hat\bbP[i\in C, \hat A_{i}]\nonumber
\end{eqnarray}
where $\hat\bbP$ denotes the law of two independent lazy random walks, starting at ${\bf 1}$ and $i$, and
 \begin{eqnarray}
\hat A_{1,i}&=&\{\tau_L({\bf 1})>\tau_B({\bf 1}),\tau_L(i)>\tau_B(i), \tau_I({\bf 1},i)>(\tau_B({\bf 1})\vee \tau_B(i))\}\nonumber\\
\hat A_1&=&\{\tau_L({\bf 1})>\tau_B({\bf 1})\}\nonumber\\
\hat A_i &=&\{\tau_L(i)>\tau_B(i)\}.\nonumber
\end{eqnarray}
Thus, we obtain that :
\begin{equation}
\bbP[(1,{\bf 1})\in S,x\in S, A_{1,x}]-\bbP[(1,{\bf 1})\in S, A_{1}]\bbP[x\in S, A_{x}]\leq 0,
\end{equation}
since $\bbP[(1,{\bf 1})\in S, A_{1}]=\hat\bbP[{\bf 1}\in C, \hat A_{1}]$ and $\bbP[x\in S, A_{x}]=\hat\bbP[i\in C, \hat A_{i}]$.
But this means that the sum over points of $I_1$ in (\ref{var1}) may be bounded by
\begin{eqnarray}
\sum_{x\in I_1}\Delta_x&\leq&\sum_{x\in I_1}\bbP[(1,{\bf 1})\in S,x\in S, A^c_{1,x}]-\bbP[(1,{\bf 1})\in S, A^c_{1}]\bbP[x\in S, A_{x}]\nonumber\\
&&-\sum_{x\in I_1}\bbP[(1,{\bf 1})\in S, A_{1}]\bbP[x\in S, A^c_{x}]\nonumber\\
&\leq& \sum_{x\in I_1}\bbP[(1,{\bf 1})\in S,x\in S, A^c_{1,x}]\nonumber\\
&\leq& \sum_{x\in I_1} \bbP[A_1^c]+\bbP[A_x^c]+\bbP[A_I^c(x)]\nonumber\\
&\leq& \sum_{x\in I_1} 2 \hat \bbP[\hat A_1^c]+\bbP[ A_I^c(x)]\label{var3}
\end{eqnarray}
where we have introduced the set :
$$
A_I(x)=\{t_I((1,\bo),x)>(t_B(1,{\bf 1})\vee t_B(x))\}.
$$
But we write 
\begin{eqnarray}
\bbP[A_I^c(x)]&=&\bbP[A_1,A_x,A_I^c(x)]+\bbP[A_1^c\cup A^c_x,A_I^c(x)]\nonumber\\
&\leq&\bbP[A_1, A_x, A_I^c(x)]+\bbP[A_1^c,A_I^c(x)]+\bbP[A_x^c,A_I^c(x)]\nonumber\\
&\leq&\hat\bbP[\tau_I(\bo,i)\leq(\tau_B(\bo)\vee\tau_B(i))]+2\hat\bbP[\hat A_1^c]\nonumber\\
\end{eqnarray}
and injecting this bound back in (\ref{var3}), we get :
$$
\sum_{x\in I_1}\Delta_x\leq\sum_{x\in I_1}4\hat \bbP[\hat A_1^c]+\hat\bbP[\tau_I(\bo,i)\leq(\tau_B(\bo)\vee\tau_B(i))].
$$
The first term is bounded using Proposition \ref{loop_exit} by $\frac C {N^{\frac 5 2}}$ while the second term is bounded using Proposition \ref{intersection} by $\frac C {N^{\frac{9}{4}}}$, because $d(\bo,i)> N^{\frac{3}{4}}$.  Thus using that $|I_1|\leq N^d$ we obtain that
$$
\frac{1}{N^{d-2}}\sum_{x\in I_1}\Delta_x \leq\frac{C}{N^{\frac{1}{4}}}
$$
which, together with (\ref{var2}) gives the result.

\end{proof}
\section{ Fick's law}
\begin{theorem}
\label{PFickmain}
Let  $d\geq 7$, $\rho_I,\rho_+,\rho_-\in (0,1)$ and  $\xi$ a family of Bernoulli random variables of parameter $\mu$ and $\{\sigma(x;0): x\in \cC\}$ be a set of independent Bernoulli random variables with $\bbE[\sigma(x;0)]=\rho_-$ if $x\in B_{-}$, $\bbE[\sigma(x;0)]=\rho_+$  if $x\in B_{+}$, and $\bbE[\sigma(x;0)]=\rho_I$ if $x\notin B_{-}\cup B_{+}$.  
\begin{enumerate}
\item For any $N\in\bbN^*$ and any $t\geq \overline t=N^{d+1}$, in law, we have $J(l,t)=J(l,\overline t):=\overline J(l)$. 
\item For any $\delta>0$ and any $l\in{1,\ldots,N-1}$,
\begin{equation}
\lim_{N\to\infty}\bbP[\left|N\bar J(l  )-\kappa(\mu)(\rho_--\rho_+)\right|>\delta]=0,
\label{Fickstationarybis}
\end{equation}
where $\kappa(\mu)=\mu(1-\mu)^{4d-2}$.
\item For any $\delta>0$, $\varepsilon>0$, $t>0$ and $l\in{1,\ldots,N-1}$,
\begin{equation}
\lim_{N\to\infty}\bbP[\left|N(J(l  ,t N^2)+L(l, t N^2))-\kappa(\mu)(\rho_--\rho_+)\right|>\delta]=0
\label{Fick}
\end{equation}
and $L(l,t)$  satisfies :

\begin{equation}
\bbP[|L(l, t N^2)|>\varepsilon]\leq C \varepsilon^{-1}\exp-\kappa(\mu) t+O(\frac{1}{N^{\frac{5}{2}}}),
\label{Fickbis}
\end{equation}
Moreover, for any $l$ and  for any $t\geq N^{d+1}$, $L(l,t)=0$.
\end{enumerate}
\end{theorem}

\begin{proof}
To show the first statement of the theorem, we use (\ref{decomporb}), (\ref{jorb}) and (\ref{dexccu}).  We observe that for any $t\geq 0$, the $\sigma$ that appear in (\ref{jorb}) are all Bernoulli variables with same parameter $\rho_I$.  Next,  if $t\geq N^{d+1}$, then all the $\sigma$ appearing in (\ref{dexccu}) are Bernoulli variables of parameter $\rho_L$ or $\rho_R$ depending on whether $x$ belongs to $B_-$ or $B_+$, because $N^{d+1}$ is the maximal size of an orbit. 

In order to obtain (\ref{Fickstationarybis}), we decompose :
\begin{eqnarray}
\bbP[\left|N\overline J(l  )-\kappa(\mu)(\rho_--\rho_+)\right|>\delta]\leq&& \bbP\left[\left|J(l,\overline t)-\frac{{\mathcal N}}{N^d}(\rho_--\rho_+)+L(l,t)\right|\geq \frac{\delta}{3N} \right]\nonumber\\
&+&\bbP[\left|\frac{{\mathcal N}}{N^{d-1}}(\rho_--\rho_+)-\kappa(\mu)(\rho_--\rho_+)\right|\geq\frac\delta 3]\nonumber \\
&+&\bbP[NL(l,\overline t)>\frac\delta  3]
\label{decompopo}
\end{eqnarray}

Proposition \ref{Pcrossing} implies that $\forall \delta>0$ and $\forall t\in\bbN^*$,
$$
\bbP\left[\left|J(l,t)-\frac{{\mathcal N}}{N^d}(\rho_--\rho_+)+L(l,t)\right|\geq \frac{\delta}{N} \right]\leq 2\exp(-\delta^2 N^{d-2}), l\in{1,\ldots,N-1}, \quad
$$
and by Proposition  \ref{crossingfick}, we already know that
\begin{equation}
\bbP[\left|\frac{{\mathcal N}}{N^{d-1}}(\rho_--\rho_+)-\kappa(\mu)(\rho_--\rho_+)\right|\geq\delta]\leq \frac {C}{\delta^2 N^{\frac 1 4}},\;\forall \delta>0.
\end{equation}
Moreover by Proposition \ref{Pcrossing}, $\forall t\geq N^{d+1}$, $L(l,t)=0$, thus we obtain (\ref{Fickstationarybis}).

\noindent In order to obtain (\ref{Fick}), we decompose :
\begin{eqnarray}
\bbP[\left|N(J(l  ,t N^2)+L(l, t N^2))-\kappa(\mu)(\rho_--\rho_+)\right|>\delta]\leq &&\bbP\left[\left|J(l,N^2t)+L(l,N^2t)-\frac{{\mathcal N}}{N^d}(\rho_--\rho_+)\right|\geq \frac{\delta}{2N} \right]\nonumber\\
&+&\bbP[\left|\frac{{\mathcal N}}{N^{d-1}}(\rho_--\rho_+)-\kappa(\mu)(\rho_--\rho_+)\right|\geq\frac\delta 2]\nonumber \\
\end{eqnarray}
Using again Proposition \ref{Pcrossing} and Proposition \ref{crossingfick} we conclude.
 Thus, we are left to show,
\begin{equation}
\bbP[ |L( l,t N^2)|\geq\varepsilon]\leq \frac{C}{\varepsilon}\exp(-\kappa(\mu) t)+O(\frac{1}{N^{\frac{5}{2}}}),\;\forall \varepsilon>0.
\label{limL}
\end{equation}
We note that by Proposition \ref{Pcrossing} and Markov inequality :
$$
\bbP[|L( l ,t N^2)|\geq\varepsilon]\leq \frac {1} {\varepsilon N^{d}}3(\bbE[{\mathcal N}_-( l,t N^2)](\rho_-+\rho_I)+\bbE[{\mathcal N}_+( l,N^2 t)](\rho_++\rho_I))
$$
and next that  :
\begin{eqnarray}
\frac 1 {N^{d}}\bbE[\cN_-(l,N^2  t)]&=&\;\bbP[(1,{\bf 1})\in S(l,t N^{2} )]\nonumber\\
&\leq& \;\bbP[F(1,{\bf 1})=(2,{\bf 1}+e_d),t_B(2,{\bf 1}+e_d)\geq t N^{2} ,t_B(2,{\bf 1}+e_d)< t_L(2,{\bf 1}+e_d)]\nonumber\\
&&+\; \bbP[t_B(2,{\bf 1}+e_d)\geq t_L(2,{\bf 1}+e_d)]\nonumber\\
&\leq& \;\hat\bbP[\tau_B({\bf 1}+e_d)\geq t N^{2} ]+\hat\bbP[\tau_L({\bf 1}+e_d)\leq \tau_B({\bf 1}+e_d)]\nonumber\\
&\leq& \;C\exp(-\kappa(\mu)  t )+O(\frac 1 {N^{\frac 5 2}}).\nonumber\\
\end{eqnarray}
by (\ref{exp_bound}).  Since the case of $\cN_+$ is analogous, we conclude.

\end{proof}

\vspace{5mm}
\noindent {\bf Acknowledgments.}
\noindent I thank Carlos Mejia-Monasterio, Justin Salez and Gordon Slade for useful discussions.
I thank the Yukawa Institute for Theoretical Physics at Kyoto University. Discussions during the YITP workshop YITP-W-13-06 on "Mathematical Statistical Physics" were useful to complete this work.
I was supported by the French ANR grant SHEPI.

\end{document}